\documentclass[letterpaper,twocolumn,10pt]{article}
\usepackage{usenix,epsfig,endnotes}

\usepackage{upgreek}
\usepackage[english]{babel}
\usepackage{blindtext}
\usepackage{amsmath}
\usepackage[capitalise,noabbrev]{cleveref}
\usepackage{gnuplottex}
\usepackage{subcaption}
\usepackage{subcaption,graphicx}
\usepackage[normalem]{ulem}
\usepackage{amsthm}
\usepackage{outlines}
\usepackage{totpages}
\usepackage{tablefootnote}
\usepackage{xurl}
\newtheorem{theorem}{Theorem}

\newcommand{\sys}[0]{FAST\xspace}

\usepackage{titlesec}
\titlespacing*{\section}
  {0pt}{1.8ex plus .5ex minus .5ex}{0.8ex plus .2ex}
\titlespacing*{\subsection}
  {0pt}{1.5ex plus .4ex minus .4ex}{0.6ex plus .2ex}
\titlespacing*{\subsubsection}
  {0pt}{1.2ex plus .3ex minus .3ex}{0.5ex plus .2ex}
  
\usepackage{enumitem}

\newenvironment{myitemize}
{\begin{itemize}[
    leftmargin=1.75em,
    topsep=0pt,
    partopsep=0pt,
    itemsep=0pt,
    parsep=0pt
]}
{\end{itemize}}

\usepackage{amsthm}
\makeatletter
\renewenvironment{proof}[1][\proofname]{\par
  \pushQED{\qed}%
  \vspace{-8pt} 
  \normalfont \itshape 
  \topsep6\p@\@plus6\p@ \trivlist
  \item[\hskip\labelsep\hskip\parindent 
        \itshape
        \scshape
        #1\@addpunct{.}]\ignorespaces
}{%
  \popQED\endtrivlist\@endpefalse
  \vspace{0pt} 
}
\makeatother
\newtheoremstyle{bolditalichead}
  {4pt} 
  {4pt} 
  {\itshape} 
  {1em} 
  {\bfseries \scshape} 
  {.} 
  {.5em} 
  {} 
\theoremstyle{bolditalichead}
\newtheorem*{theorem*}{Theorem}

\begin{document}
\title{FAST: An Efficient Scheduler for All-to-All GPU Communication}


\newcommand{\afflogo}[1]{\raisebox{0em}{\includegraphics[height=1.6ex]{#1}}}
\newcommand{\afflogoo}[1]{\raisebox{0em}{\includegraphics[height=1.3ex]{#1}}}
\newcommand{\afflogooo}[1]{\raisebox{0em}{\includegraphics[height=1.7ex]{#1}}}

\newcommand{\CMU}{\afflogo{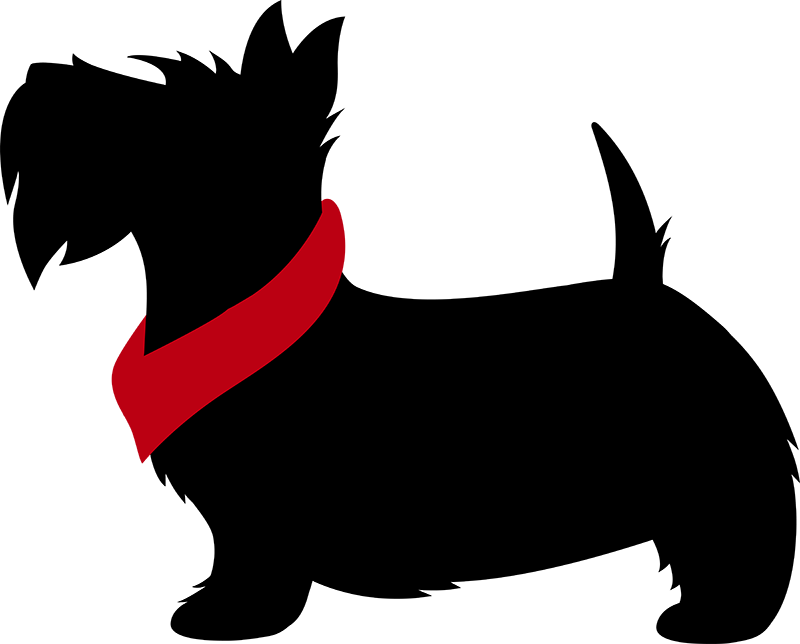}}
\newcommand{\UPENN}{\afflogo{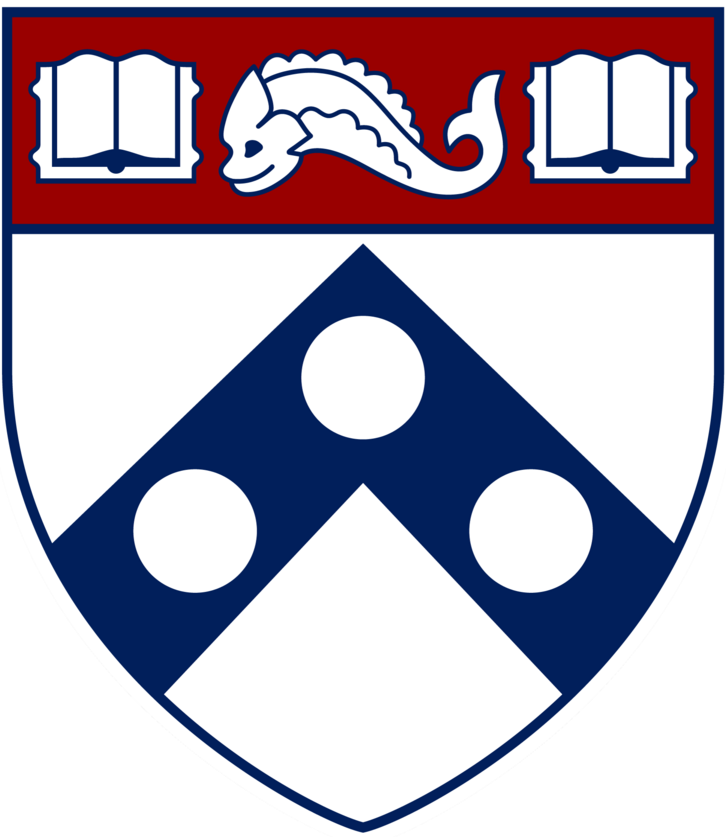}}
\newcommand{\UW}{\afflogoo{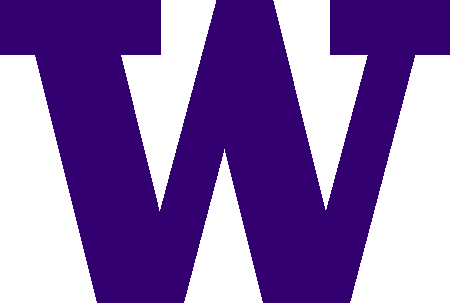}}
\newcommand{\MANGOBOOST}{\afflogooo{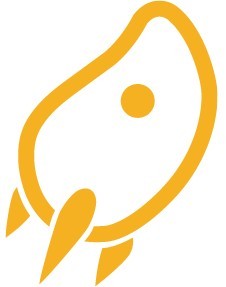}}

\author{
\begin{tabular}{cccc}
Yiran Lei\,\CMU \, \MANGOBOOST &
Dongjoo Lee\,\MANGOBOOST &
Liangyu Zhao\,\UW &
Daniar Kurniawan\,\MANGOBOOST \\
Chanmyeong Kim\,\MANGOBOOST &
Heetaek Jeong\,\MANGOBOOST &
Changsu Kim\,\MANGOBOOST &
Hyeonseong Choi\,\MANGOBOOST \\
Liangcheng Yu\,\UPENN &
Arvind Krishnamurthy\,\UW &
Justine Sherry\,\CMU &
Eriko Nurvitadhi\,\MANGOBOOST \\
\\[-0.25em]
\multicolumn{4}{c}{
{\small
\CMU\,Carnegie Mellon University \quad
\MANGOBOOST\,MangoBoost \quad
\UW\,University of Washington \quad
\UPENN\,University of Pennsylvania
}
}
\end{tabular}
}

\newcommand{\jpara}[1]{%
  \vspace{1.5pt}%
  \noindent\textbf{{#1}}.
}

\setlength{\parskip}{0pt}

\newcommand{\eg}[0]{\textit{e.g.,}\xspace}
\newcommand{\ie}[0]{\textit{i.e.,}\xspace}
\newcommand{\ata}[0]{All-to-All\xspace}
\newcommand{\atav}[0]{\texttt{alltoallv}\xspace}

\maketitle

\begin{abstract}
All-to-All(v) communication is a critical primitive in modern machine learning workloads, particularly mixture-of-experts (MoE) models.
Unfortunately, efficient scheduling is challenging due to workload skew, heterogeneous two-tier fabrics, and incast congestion, compounded by the dynamic nature of MoE workloads, where traffic shifts every few hundred milliseconds.
Existing schedulers are hardly scalable, incurring seconds to hours of synthesis time, making them impractical.

We present \sys, an efficient All-to-All(v) scheduler.
\sys addresses skew through intra-server rebalancing and enforces balanced, one-to-one scale-out transfers that avoid incast.
Evaluated extensively on both NVIDIA H200 and AMD MI300X clusters, \sys consistently outperforms state-of-the-art solutions on skewed workloads while reducing synthesis time by orders of magnitude.
\end{abstract}

\section{Introduction}
\label{sec: introduction}

\ata communication---where every endpoint sends data to all others---has long been a fundamental collective communication primitive in scientific and parallel computing, supporting workloads such as 3D FFTs~\cite{fft}. 
Early systems typically treated each server as the communication endpoint, interconnected by networks with relatively uniform bandwidth. In modern ML clusters, the endpoint has shifted to individual GPUs, which are connected through a two-tier fabric consisting of faster intra-server (scale-up) links and slower inter-server (scale-out) links.\footnote{The terms `scale-up' and `intra-server' network are used interchangeably, and likewise `scale-out' and `inter-server'.} 
Consequently, \ata has become a key operation for many ML applications, including recommendation models~\cite{dlrm1,dlrm2}, Gaussian Splatting~\cite{gaussiansplatting}, and mixture-of-experts (MoE) models~\cite{deepspeed_moe, gshard, deepseek, tutel}.

The role of \ata in ML applications is important.
In mixture-of-experts (MoE) models in particular, its cost can account for a large fraction of training time.
MoE improves the efficiency of model parameters by activating only a subset of experts for each input token rather than all simultaneously.
This selectivity, however, necessitates frequent \ata operations to dispatch tokens to experts across GPUs and aggregate their outputs.
Prior studies~\cite{mixnet, tutel} show that MoE \ata can consume 30–55\% of training time, making it a major contributor to overhead in large-scale training.

The challenges of \ata arise at both the application and system layers. 
At the application layer, traffic is often {\bf skewed} and {\bf dynamic}. 
In MoE, some experts are selected more frequently than others, leading to larger data transfers for their corresponding GPUs during \ata. 
This imbalance keeps certain GPUs and NICs busy after others have finished, creating straggler effects. 
When communication volumes differ across endpoints in \ata, the operation is referred to as \atav~\cite{alltoallv}. 
Further, the traffic pattern in \atav changes every few hundred milliseconds, as the MoE gating function reassigns tokens to experts at runtime (\cref{fig:moe}). 
As a result, a GPU that is a hotspot at one moment may be idle the next. 
This dynamism makes static schedules impractical and requires schedulers to adapt in real time.

At the system layer, the hardware fabric that connects GPUs is inherently two-tiered: fast intra-server links (scale-up) and much slower inter-server links (scale-out) (\Cref{fig:structure_intra_inter}).
This \textbf{heterogeneity} means that flows of the same size can finish quickly inside a server but take an order of magnitude longer across servers, leaving schedulers to coordinate thousands of flows over mismatched bandwidths.
In addition, \atav's dense communication pattern naturally triggers \textbf{incast}, a classic networking problem where many senders overload downlink of the same receiver.
Incast causes network congestion with switch queue buildup and reduced goodput---even under modern congestion control---and  remains an open challenge in large-scale cluster networking~\cite{ultraethernet-whitepaper, ultraethernet}.

Together, these challenges make efficient \atav scheduling especially difficult under the tight timescales demanded by MoE. 
Existing schedulers such as TACCL~\cite{taccl} and TE-CCL~\cite{teccl} employ solvers~\cite{gurobi} to generate near-optimal schedules.
While they overcome the inefficiencies of fixed schedules in collective communication libraries such as NCCL and RCCL, their resulting formulations are NP-hard~\cite{taccl} and can take minutes to hours to synthesize a schedule even for just 32 GPUs.
The state-of-the-art scheduler SyCCL~\cite{syccl} accelerates this process with heuristics and parallelism, yet requires seconds to minutes for balanced \ata, while skewed \atav still remains unresolved.
Consequently, these schedulers are far too slow for MoE \atav workloads that shift every few hundred milliseconds.

In this paper, rather than increasing scheduler complexity, we simplify the problem itself:
{\it It suffices to focus on optimizing the scale-out tier---the real bottleneck}.
We observe that scale-up is roughly an order-of-magnitude faster than scale-out (\cref{fig:intra_inter}).
So the much faster scale-up fabric can cheaply absorb skew within each server, reshaping traffic before it reaches scale-out.
With this, we can then maximize scale-out efficiency by keeping bottleneck servers transmitting at full rate while avoiding incast. 
Achieving this requires pairing senders and receivers without contention, which reduces naturally to a one-to-one matching problem between endpoints, solvable in polynomial time.

Building on this insight, we propose \sys, a polynomial-time, matching-based scheduler for skewed and dynamic \atav workloads.
\sys operates in two phases: (i) \emph{Skew mitigation}, where it uses fast scale-up fabric to rebalance the scale-out workload so that all NICs face equal volume before traversing the slow scale-out links; and (ii) \emph{Balanced, one-to-one transfers}, which use Birkhoff’s decomposition~\cite{birkhoff} to generate successive matchings between senders and receivers, ensuring scale-out transfers proceed without incast while keeping bottleneck servers fully active at line rate until completion.
While Birkhoff's decomposition has appeared in the design of switches~\cite{birkhoff_switch1, birkhoff_switch2, chronos}, to our knowledge, this is the first work to apply Birkhoff’s decomposition to scheduling collective communication at GPU endpoints.

\begin{figure}[t]
\centering
\includegraphics[width=0.66\linewidth]{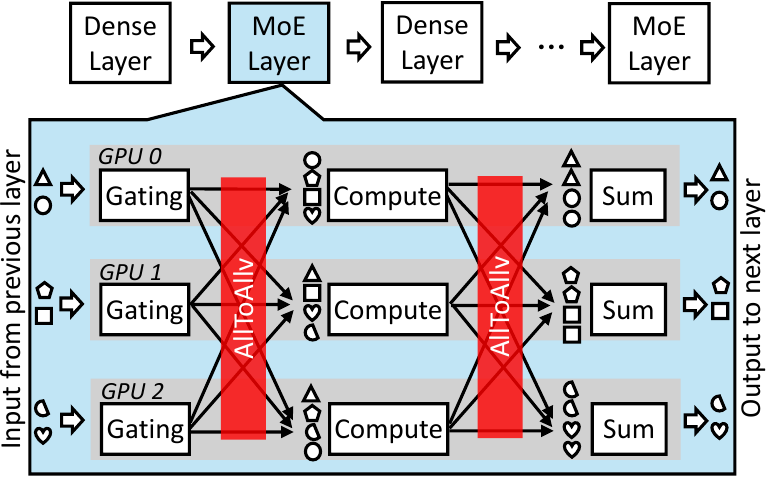}
\caption{MoE models invoke \atav twice per MoE layer, making it a critical communication primitive.}
\label{fig:moe}
\end{figure}

We implement \sys on both NVIDIA H200~\cite{h200} and AMD MI300X~\cite{mi300x} testbeds and evaluate it against state-of-the-art solutions such as DeepEP~\cite{deepep}, TACCL~\cite{taccl}, and TE-CCL~\cite{teccl}. 
Under skewed workloads, \sys outperforms the strongest NVIDIA baseline by 1.01–1.3$\times$ and the strongest AMD baselines by 1.5–2.8$\times$, and when integrated into Megatron-LM~\cite{megatron-lm} on AMD, improves end-to-end MoE training throughput by $4.48\times$ over RCCL~\cite{rccl} that suffers heavily from incast. 
The scheduler is highly efficient, completing in 221 $\upmu$s for 64 GPUs---fast enough for MoE workloads where the traffic matrix changes every few hundred milliseconds. 
Our code is available at~\cite{fast_github}.

\section{Motivation}
\label{sec:motivation}

\ata (\atav) communication---where every endpoint exchanges distinct data with all others—has become a key primitive in modern GPU clusters. In this setting, the communication endpoint is each individual GPU, connected by a two-tier fabric consisting of fast intra-server (scale-up) links and slower inter-server (scale-out) links.

\jpara{\ata communication cost in MoE models}
Mixture-of-experts (MoE) is now a leading architecture for scalable large language models: instead of activating the full model for every input token, a lightweight gating network selects only a subset of `experts', each implemented as a feed-forward network (FFN).
While expert parallelism (EP) improves parameter efficiency, it introduces frequent, large-scale \atav operations to dispatch tokens to selected experts---often spanning hundreds of GPUs~\cite{perplexityai, deepseek}---and to gather results.
As illustrated in \Cref{fig:moe}, each MoE layer invokes \atav twice, and MoE layers often constitute a large fraction of the model.
Prior measurements~\cite{tutel, mixnet} have shown that \atav can account for 30–56\% of training time.

While MoE is our primary focus, the importance of \atav extends beyond. It underpins recommendation systems~\cite{dlrm1,dlrm2}, Gaussian Splatting~\cite{gaussiansplatting}, and classical scientific workloads such as 3D FFT~\cite{fft} (where it can dominate up to 97.3\% of runtime~\cite{fftoverhead}!). Despite being just one collective, \atav performance disproportionately shapes the efficiency of both modern AI and traditional HPC workloads.
\begin{figure}
\centering
\begin{subfigure}[b]{0.22\textwidth}
\includegraphics[width=\linewidth, trim=7 1 7 5, clip]{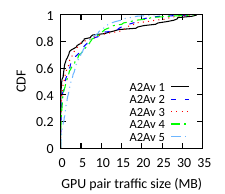}
\caption{Skewness}
\label{fig:workload_cdf}
\end{subfigure}\hfill
\begin{subfigure}[b]{0.22\textwidth}
\includegraphics[width=\linewidth, trim=7 1 5 3, clip]{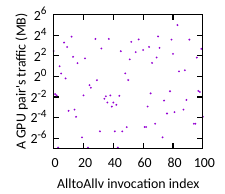}
\caption{Dynamism}
\label{fig:workload_sequence}
\end{subfigure}
\caption{\ata workload is skewed and dynamic when using Megatron-LM to pre-train a MoE model.}
\label{fig: megatron_workload}
\end{figure}

\jpara{Application challenges---skewness \& dynamism}
To study this communication, we profile MoE training using Megatron-LM~\cite{megatron-lm} with 32 experts (one per GPU).
We find that \atav workloads are inherently {\it skewed} and {\it dynamic}, consistent with recent profiling results of Mixtral models~\cite{mixnet, mixtral}.
Unlike balanced collectives such as All-Reduce, MoE \atav generates a highly uneven demand matrix: some GPU pairs exchange more than 12$\times$ the median volume (\Cref{fig:workload_cdf}).
This skew creates \textit{stragglers}---NICs that remain busy long after others have finished, delaying the entire collective and stalling training progress (\Cref{fig:straggler}).
In this setting, scheduling is critical: by routing part of a heavy GPU’s traffic to idle NICs, a smart scheduler can smooth out skew and mitigate stragglers, as illustrated in~\cref{fig:straggler_mitigation}.

\begin{figure}[t]
\centering
\begin{subfigure}[b]{\linewidth}
\centering
\includegraphics[width=0.65\linewidth]{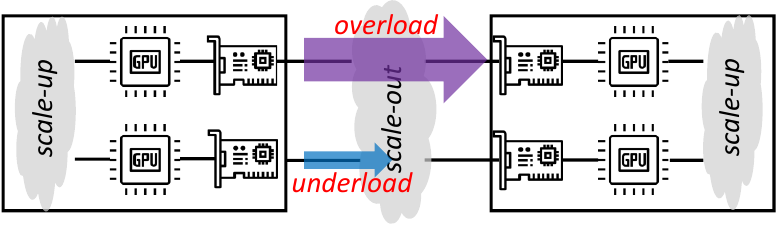}
\caption{Stragglers from skewed workload}
\label{fig:straggler}
\end{subfigure}
\begin{subfigure}[b]{\linewidth}
\centering
\includegraphics[width=0.65\linewidth]{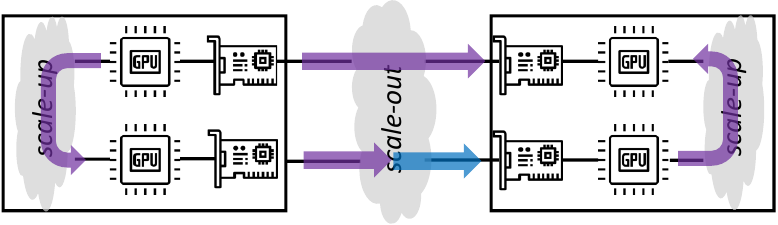}
\caption{Mitigation via rebalanced traffic}
\label{fig:straggler_mitigation}
\end{subfigure}
\caption{Workload skew creates stragglers and underutilized NICs that degrade performance, which can be mitigated through traffic rebalancing.}
\label{fig:straggler_and_mitigation}
\end{figure}

MoE workloads are also \emph{dynamic}: \atav traffic pattern changes every few hundred milliseconds. 
As shown in \Cref{fig:workload_sequence}, the traffic volume of a given GPU pair can vary significantly across successive \atav invocations, since token routing is jointly decided by the input tokens and per-MoE-layer gating functions (\Cref{fig:moe}) and cannot be predicted in advance. 
Therefore, communication schedules must be recomputed online at the timescale as workload changes---a GPU pair with heavy load in one \atav may be nearly idle in the next---making fast, online scheduling essential.

\jpara{System challenges---heterogeneity \& incast}
Modern GPU clusters further complicate scheduling with two system-level challenges.
First, {\it a heterogeneous, two-tier fabric} connects GPUs (\Cref{fig:structure}): fast intra-server links (e.g., 5th-gen NVLink 900 GBps~\cite{h100}) and much slower inter-server links (e.g., Ethernet 800 Gbps~\cite{cx8}).
This two-tier network increases scheduling complexity: for each \atav, schedulers must navigate thousands of flows across mismatched bandwidths, exploring a large space of routing and pacing decisions, which can turn scheduling itself into a bottleneck.

Second, {\it incast} is a classic networking problem that arises from \atav's dense communication pattern, where many flows converge on the same NIC downlink in the scale-out fabric. 
While the burstiness of small messages can be absorbed by switch queues, MoE \atav transfers are much larger---typically 100 MB to 1 GB~\cite{deepep}---causing sustained congestion that requires active control. 
Even with advanced schemes~\cite{dcqcn, timely, ultraethernet}, incast often results in unfair bandwidth sharing and degraded goodput. 
It remains an open challenge~\cite{ultraethernet-whitepaper}, and schedulers often attempt to mitigate it proactively rather than relying on the transport layer alone.

\begin{figure}[t]
\centering
\begin{subfigure}[b]{0.2\textwidth}
\centering
\includegraphics[width=\linewidth]{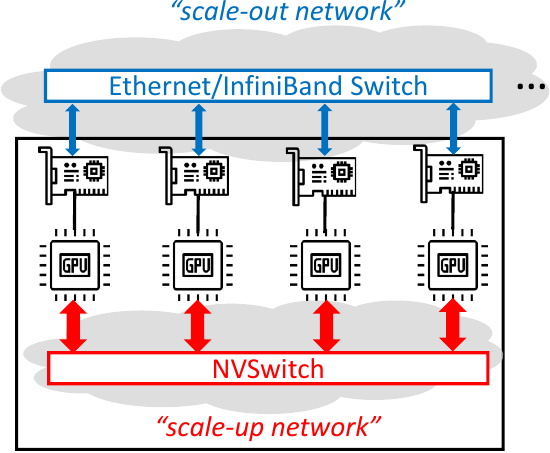}
\caption{Two-tier structure}
\label{fig:structure}
\end{subfigure}
\begin{subfigure}[b]{0.25\textwidth}
\centering
\includegraphics[width=0.875\linewidth, trim=7 1 7 4, clip]{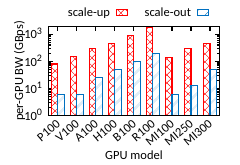}
\caption{Per-GPU full-duplex bandwidth}
\label{fig:intra_inter}
\end{subfigure}
\caption{Modern GPU clusters feature a two-tier fabric: a high-bandwidth intra-server scale-up network (e.g., NVLink, Infinity Fabric) and a lower-bandwidth inter-server scale-out network (e.g., Ethernet, InfiniBand).}
\label{fig:structure_intra_inter}
\end{figure}

\jpara{Limitations of existing approaches}
State-of-the-art schedulers such as TACCL~\cite{taccl}, TE-CCL~\cite{teccl}, and SyCCL~\cite{syccl} are designed to be {\it general}: they support a wide range of collectives (All-Reduce, All-Gather, All-to-All, etc.) and reason about arbitrary topologies.
To achieve this generality, they cast scheduling as constraint-satisfaction or optimization problems, often NP-hard~\cite{taccl}, and solve them with heavy computation.
This is practical for collectives with {\it repetitive} communication patterns like All-Reduce, where the high scheduling cost can be amortized over many iterations.

For \atav, however, existing approaches are too slow. 
SyCCL~\cite{syccl}, the fastest to date with parallelism and heuristic acceleration, is orders of magnitude faster than earlier solver-based systems but still requires minutes to produce a schedule for 64 GPUs---impractical when MoE traffic shifts every few hundred milliseconds. 
While such systems excel at achieving near-optimal completion under arbitrary topologies, their fine-grained modeling makes them hardly scalable to this setting.

At the other extreme, production libraries like NCCL~\cite{nccl} generate schedules instantly, but rely on fixed schedules oblivious to the dynamic, skewed workload, often resulting in lower throughput than what the hardware could achieve.

\jpara{Goal}
Can we design a fast, online scheduler for \atav that sustains high performance? Instead of targeting arbitrary collectives or topologies, we focus on a specialized solution for \atav on today’s two-tier GPU clusters---where skewness, dynamism, asymmetry, and incast constrain existing schedulers. FAST integrates with existing libraries: the runtime dispatches \atav to FAST and uses conventional algorithms for other collectives.

\section{Design Overview}
\label{sec: design_overview}

To generate a separate schedule for each \atav invocation in modern GPU clusters, we take a step back.
Rather than enumerating complex constraints, we start by first solving \atav on a simplified single-tier network and then generalize the solution to today's asymmetric two-tier fabrics.

\begin{figure}[t]
\centering
\includegraphics[width=\linewidth]{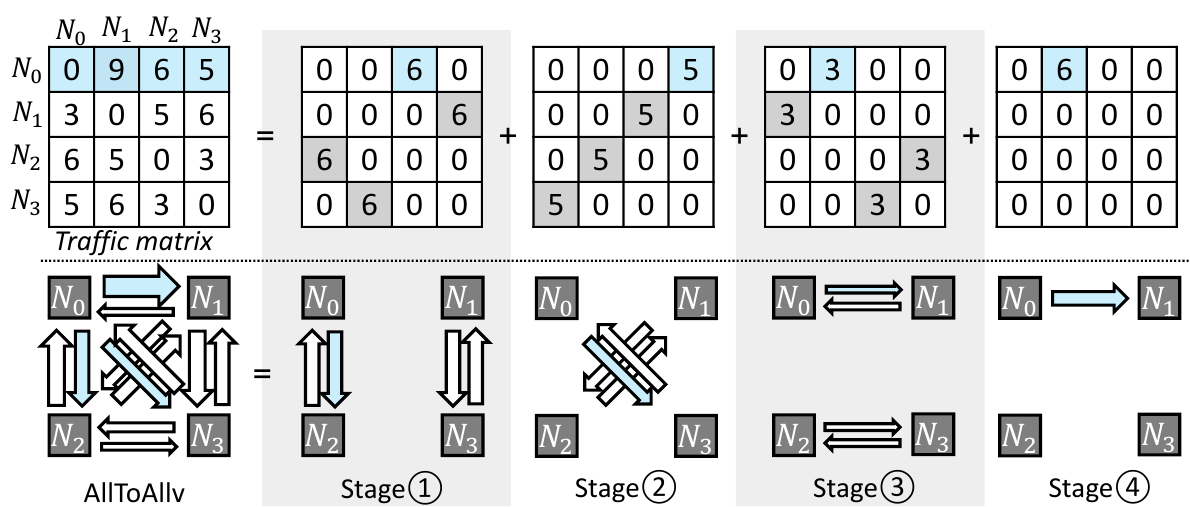}
\caption{Birkhoff's decomposition of a 4-node \atav.  
Completion time is dictated by the largest sender ($N_0$ in blue), and the schedule is optimal since $N_0$ stays active in every stage while lighter nodes drop out early.}
\label{fig:bikrhoff_example}
\end{figure}

\jpara{Starting point: \atav in a single-tier network}  
On a single-tier, full-bisection network with uniform link bandwidth, the objective is to maximize communication efficiency by avoiding incast and congestion. 
This is achieved by ensuring that, at any instant, each sender communicates with exactly one receiver and each receiver accepts data from exactly one sender. 
As a result, communication can be organized into stages, where each stage realizes a one-to-one matching between senders and receivers, and successive stages collectively exchange data across all sender–receiver pairs.

For such a schedule to reach the theoretical minimum completion time, two conditions must hold:  
(i) each stage is balanced so all active nodes start and finish together, and  
(ii) the bottleneck endpoints (the heaviest senders or receivers) remain fully active at line rate until completion.  

To solve this problem, we observe that Birkhoff’s decomposition~\cite{birkhoff}---introduced in 1946 as a mathematical theorem---can be reinterpreted as an optimal scheduling strategy for \atav.
Formally, the theorem states that any traffic matrix\footnote{The theorem applies to scaled doubly stochastic matrices; arbitrary matrices can be adapted, as described later.} can be expressed as a weighted sum of \emph{permutation matrices}.
Viewed through the scheduling lens, each permutation corresponds to a transfer stage: every active row (sender) and column (receiver) has exactly one nonzero entry of equal size (transfer size), so each participant exchanges data with exactly one partner and all finish the stage together.
By summing over these permutation matrices, all flows advance in proportion to their demand, driving the transfer to completion.

\Cref{fig:bikrhoff_example} illustrates our strawman approach: the top pane shows the decomposition into (partial) permutation matrices, while the bottom pane shows the corresponding transfers.  
This approach is appealing because 
(i) completion time hits the lower bound---the bottleneck node (e.g., $N_0$ as sender in blue) transmits in \emph{all} stages;
(ii) each stage is balanced until nodes finish, so participants advance to the next stage together; 
and (iii) the decomposition is computationally efficient.

\jpara{Challenges of applying Birkhoff's decomposition to two-tier networks}
Unfortunately, modern GPU clusters deviate from the simplified network setting in two ways.
(i) In a two-tier fabric, even a `balanced' permutation stage can finish unevenly---transfers on the inter-server links lag behind, idling the faster intra-server links and wasting bandwidth.  
(ii) In 8-GPU-per-server (e.g., HGX~\cite{hgx}) clusters, completion time is dictated by the busiest inter-server endpoints. Because most GPU pairs span server, the faster intra-server tier is rarely limiting. Under heavy skew, Birkhoff-based schedules cannot relieve this inter-server bottleneck and therefore still stall.

\jpara{Our approach}
We exploit the two-tier fabric as an opportunity to simplify scheduling. 
Since intra-server (scale-up) bandwidth is far higher than inter-server (scale-out), it is rarely a bottleneck. 
We instead repurpose it to \emph{reshape traffic in advance}, absorbing imbalance locally so the scale-out tier sees a more uniform workload and Birkhoff’s decomposition can produce more efficient schedules.

At the heart of this reshaping is another simple yet powerful observation: \emph{what matters is delivering data to the right servers}.
Which GPU inside a server handles the transfer is secondary, since intra-server shuffles are cheap compared to scale-out communication.
This leads to a two-phase design:

\begin{figure}
\centering
\includegraphics[width=0.9\linewidth]{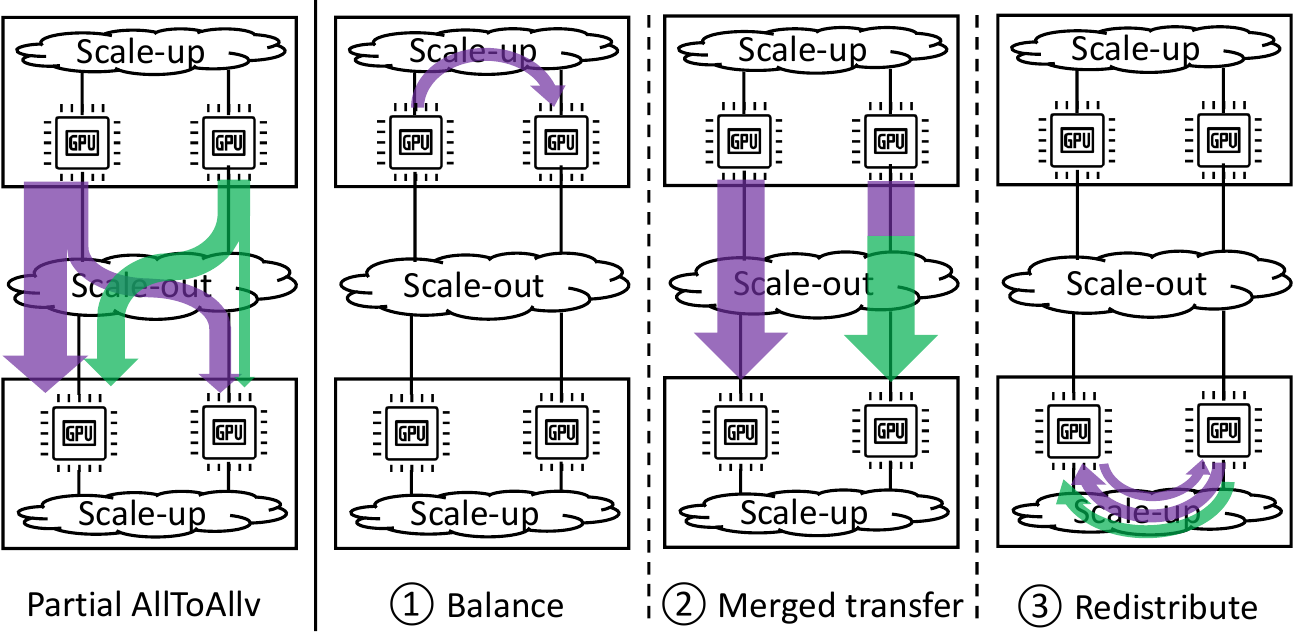}
\caption{\sys: transforming an \atav workload with sender/receiver stragglers into a balanced scale-out transfer. Each GPU is connected to a dedicated NIC (omitted).}
\label{fig:phase1_design}
\end{figure}

\emph{Intra-server scheduling: balancing and redistribution (\S\ref{subsec: phase1_design}).}
Within each server, we equalize traffic across GPUs before it leaves the node (\cref{fig:phase1_design}): overloaded GPUs hand off excess traffic to lighter ones so every NIC carries the same volume per destination server.
On the receiving side, each GPU receives data from exactly one designated sender on each source server, equalizing incoming volume.
This process makes some data initially arrive at a `proxy' GPU at the correct destination server, which is then quickly forwarded to the true destination GPU via a cheap intra-server redistribution.

\emph{Inter-server scheduling: balanced one-to-one transfers (\S\ref{subsec: phase2_design}).}
Once intra-server skew is absorbed, the remaining challenges are server-level imbalance and incast.
Here, we apply Birkhoff’s decomposition to construct successive one-to-one, balanced transfer stages, ensuring bottleneck servers remain active at line rate until their traffic is complete.

Finally, we pipeline the two phases to tighten end-to-end transfer (\S\ref{subsec: end_to_to_pipeline}), and conclude by analyzing key properties (e.g., optimality and complexity) (\S\ref{subsec: properties}).

\section{Two Phase Scheduler}
\label{sec: design}

This section presents our design: mitigating skew within servers (\S\ref{subsec: phase1_design}), handling incast and imbalance across servers (\S\ref{subsec: phase2_design}), and pipelining both to improve end-to-end transfers (\S\ref{subsec: end_to_to_pipeline}). We conclude with an analysis of key properties (\S\ref{subsec: properties}).

\subsection{Intra-server Scheduling: Balancing and Redistribution}
\label{subsec: phase1_design}

The first challenge arises within each server: GPUs often generate or absorb uneven volumes of traffic, creating stragglers where some NICs sit idle while others are overloaded.
By leveraging the fast scale-up fabric, our goal is to eliminate these sender- and receiver-side imbalances, presenting the scale-out tier with a more uniform workload.

\jpara{Example setup}  
Consider a simple 2-server case ($A$, $B$), each with 2 GPUs ($A_0$, $A_1$, $B_0$, $B_1$).  
The workload is represented by a 4$\times$4 GPU-to-GPU traffic matrix (\Cref{fig:matrix_transformation_process}).  
Each \emph{row sum} reflects a GPU’s total outgoing volume; each \emph{column sum} reflects its total incoming volume. 
Cross-server transfers appear as 2$\times$2 tiles (blue for $A\!\rightarrow\!B$, green for $B\!\rightarrow\!A$).   
Since scale-out is the bottleneck, we focus on these tiles, omitting the grey intra-server diagonals for the purpose of illustration.

\begin{figure}
\centering
\includegraphics[width=\linewidth]{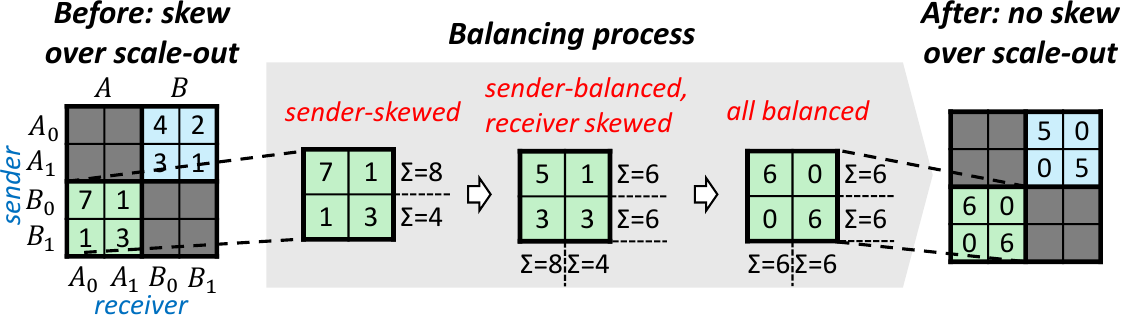}
\caption{Balancing process for a 2-server, 2-GPU-per-server \atav: a skewed tile (left) is reshaped into a scalar form (right), ensuring no GPU NIC is overloaded. Grey tiles (intra-server) omitted for clarity.
}
\label{fig:matrix_transformation_process}
\end{figure}

\jpara{Mitigating sender skew}  
The first step is to prevent a GPU from being the `straggler sender'.
In the $B\!\rightarrow\!A$ tile, GPU $B_0$ must send 8 units, while $B_1$ only needs 4.  
If transmitted directly, $B_1$ would finish early, leaving $B_0$ as a straggler.  
To avoid this, we rebalance within server $B$: heavily loaded GPUs shift part of their traffic to lightly loaded ones using the scale-up fabric.  
Here, $B_0$ transfers 2 units to $B_1$, so both end up with 6.  
In matrix terms, the \emph{row sums} of the tile are equalized---ensuring every NIC in $B$ contributes the same total outgoing load to server $A$.  

\jpara{Mitigating receiver skew}  
After sender-side balancing, GPU $A_0$ may still receive 8 units (column sum) while $A_1$ only needs 4, leaving a receiver-side straggler.
The fix is to decouple the notions of `correct server' and `correct GPU'.
Each sender forwards \emph{all} of its traffic to its peer GPU with the same local index ($B_0 \rightarrow A_0$, $B_1 \rightarrow A_1$), ensuring data first arrives at the correct server.
This \emph{merged peer transfer} keeps receiver loads balanced---since senders were equalized earlier---even though some data arrives temporarily at the wrong GPU, to be corrected later.
In matrix form, each row collapses into a single nonzero in its row-local-index column, turning the $2\times2$ tile into a \emph{scalar matrix}  with equal diagonal entries and all off-diagonals zero (right in~\cref{fig:matrix_transformation_process}).  
The result is one-to-one, balanced scale-out transfers across GPUs.

\jpara{Redistribution}  
At this point, all traffic reaches the correct \emph{server}, but may still be at the {\it wrong GPU}.  
A final redistribution step corrects placement inside the server, routing traffic from the proxy GPU to the true destination over the scale-up fabric.  
Because scale-up is an order of magnitude faster than scale-out, this added step incurs small overhead.

By combining sender balancing, merged peer transfers, and local redistribution, our scheduler reshapes {\it each cross-server tile} into its most balanced \emph{scalar form}.
This transformation removes intra-server skew and equalizes scale-out traffic at GPU level.
What remains is the higher-level server-to-server skew, which we address next.

\begin{figure}
\centering
\includegraphics[width=0.75\linewidth]{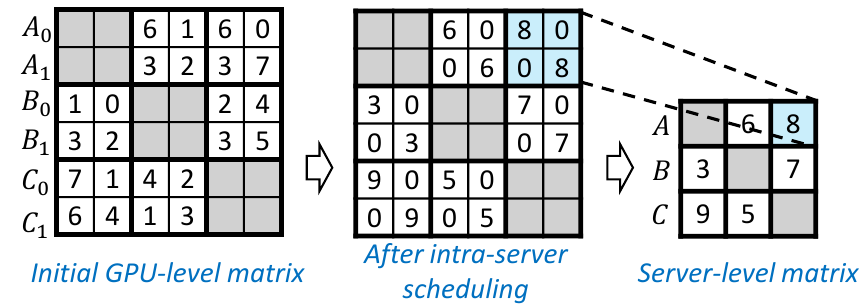}
\caption{Intra-server scheduling reduces a 6×6 GPU-level matrix ($A_0,A_1$,$\cdots$,$C_1$) to a 3×3 server-level matrix ($A,B,C$).}

\label{fig:server_level_matrix}
\end{figure}

\subsection{Inter-server Scheduling: Balanced, One-to-one Transfers}
\label{subsec: phase2_design}

Although intra-server scheduling reshapes the initial skewed GPU-level \atav into a more balanced form, balancing over scale-up network \emph{within each server} cannot eliminate \emph{server-to-server} skew.
Some servers still send or receive more traffic than others, creating bottlenecks that must be addressed to achieve good end-to-end performance.

This imbalance becomes clearer once we reduce the GPU-to-GPU traffic matrix into a simpler server-to-server view.
\Cref{fig:server_level_matrix} shows a 3-server, 2-GPU-per-server example: the original 6×6 GPU matrix (left) is reshaped by intra-server scheduling into the balanced form (middle), where each 2$\times$2 server-to-server tile becomes a scalar matrix.
Each scalar tile can then be collapsed into a single entry, yielding the reduced 3$\times$3 server-level matrix (right).
The intuition is that, after intra-server scheduling, GPUs within a server {\it act identically over scale-out}---each sending and receiving equal volumes---so we can abstract away individual GPUs.
This reduction both exposes the remaining skew across servers and simplifies scheduling, since the server-level problem is typically an order of magnitude smaller than the GPU-level one (e.g., with 8 GPUs per server~\cite{hgx} in modern clusters).

At this server level, two scheduling challenges remain:
(i) incast: even with peer access, GPU$_j$ from multiple servers may funnel into GPU$_j$ of the same destination server, overloading its scale-out link; and
(ii) throughput optimality: the busiest servers must remain fully active at line rate until their traffic flows complete, otherwise the overall completion time lags behind the achievable minimum.

\begin{figure}
\centering
\includegraphics[width=\linewidth]{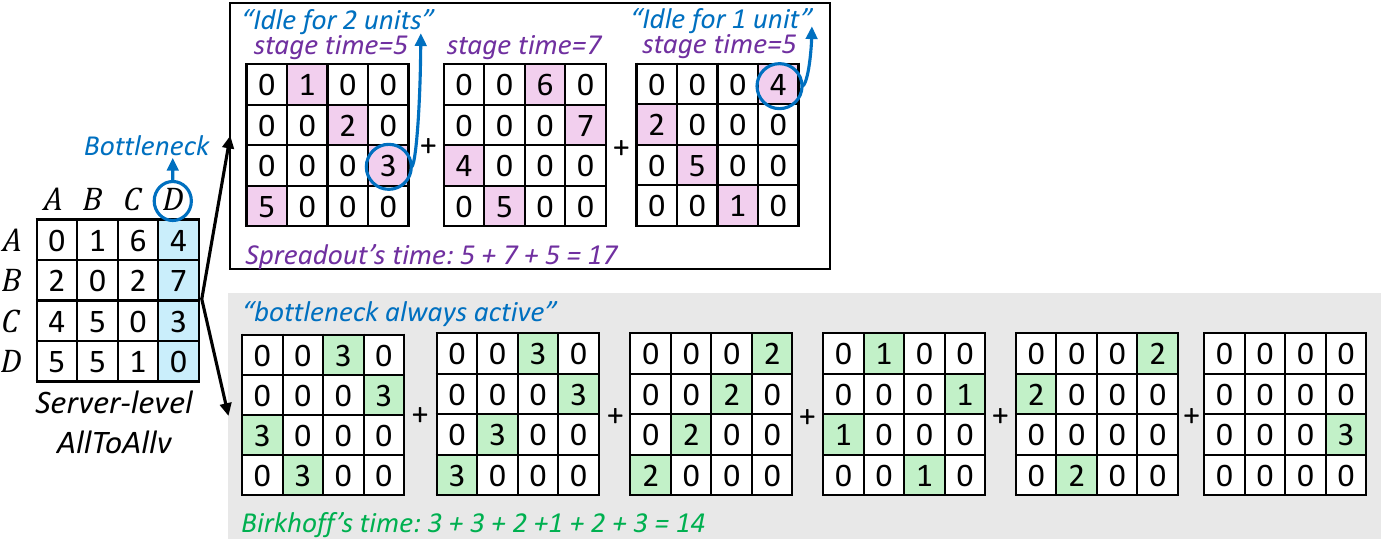}
\caption{SpreadOut vs.\ Birkhoff. Both use one-to-one mappings, but SpreadOut (top) can stall at each stage by leaving the bottleneck server idle, while Birkhoff (bottom) keeps it continuously active until completion (optimal).}
\label{fig:birkhoff_vs_spreadout}
\end{figure}

\jpara{SpreadOut: one-to-one but not optimal}  
A natural incast-free baseline is MPI’s SpreadOut algorithm~\cite{mpi_all2all}, which cycles through `shifted diagonals' of the $N \times N$ server matrix: at stage $i$, server $s$ sends to server $(s+i)\%N$.  
This guarantees one-to-one sender-receiver mappings at every stage.  

However, SpreadOut fails the second requirement: the bottleneck server may sit idle during many stages.  
Although the bottleneck is the row or column with the largest \emph{sum}, the entry selected for that row or column in a given stage need not be the largest entry on the corresponding diagonal.
When this happens, the stage is gated by another matrix entry (i.e., a flow from non-bottlenecks), forcing the true bottleneck to wait.

\Cref{fig:birkhoff_vs_spreadout} shows an example: 
server~$D$, {\it as a receiver}, is the bottleneck---with 14 units column sum---the heaviest among all rows/columns.  
But in stage~1, $D$ receives only 3 units, while the stage is gated by a separate 5-unit flow ($D \rightarrow A$ where $D$ is {\it a sender}).
As a result, $D$ sits idle {\it as a receiver} for 2 extra time units.
The same situation recurs in stage~3, adding another 1 unit of idle time.
Altogether, SpreadOut finishes in 17 units---3 units slower than the 14-unit theoretical minimum.  

In matrix terms, SpreadOut’s completion time equals the \emph{sum of the maximum entry on each diagonal}.  
This sum is provably no smaller than the largest row or column sum---the true lower bound---so SpreadOut can not guarantee optimality.  

\jpara{Birkhoff's decomposition: one-to-one with optimality}  
The optimal completion time is determined by the busiest server---the largest row or column sum in the matrix.
In \Cref{fig:birkhoff_vs_spreadout}, the busiest server~$D$ must receive 14 units, so the minimum possible time is 14.
Hitting this bound requires $D$ to receive at line rate in every stage.

Birkhoff's decomposition~\cite{birkhoff} is tailored to this setting.  
It expresses a traffic matrix as a weighted sum of \emph{permutation matrices}.  
Each permutation matrix has exactly one nonzero per row and column, all of identical value, representing a one-to-one, balanced transfer stage---every active sender transmits the same amount to exactly one receiver, and each receiver accepts from one sender.  
Some permutation matrices may be \emph{partial}, with zero rows or columns for servers that have already finished.
In \Cref{fig:birkhoff_vs_spreadout}, the first four stages are full permutation matrices, while the last two are partial.  

Viewed as a schedule, this decomposition yields a sequence of one-to-one sender-receiver matchings where bottleneck servers remain continuously active until they complete.  
In our example, Birkhoff finishes in 14 units---exactly the lower bound---achieving optimality.  

\begin{figure*}[t]
  \centering
  \includegraphics[width=1.0\linewidth]{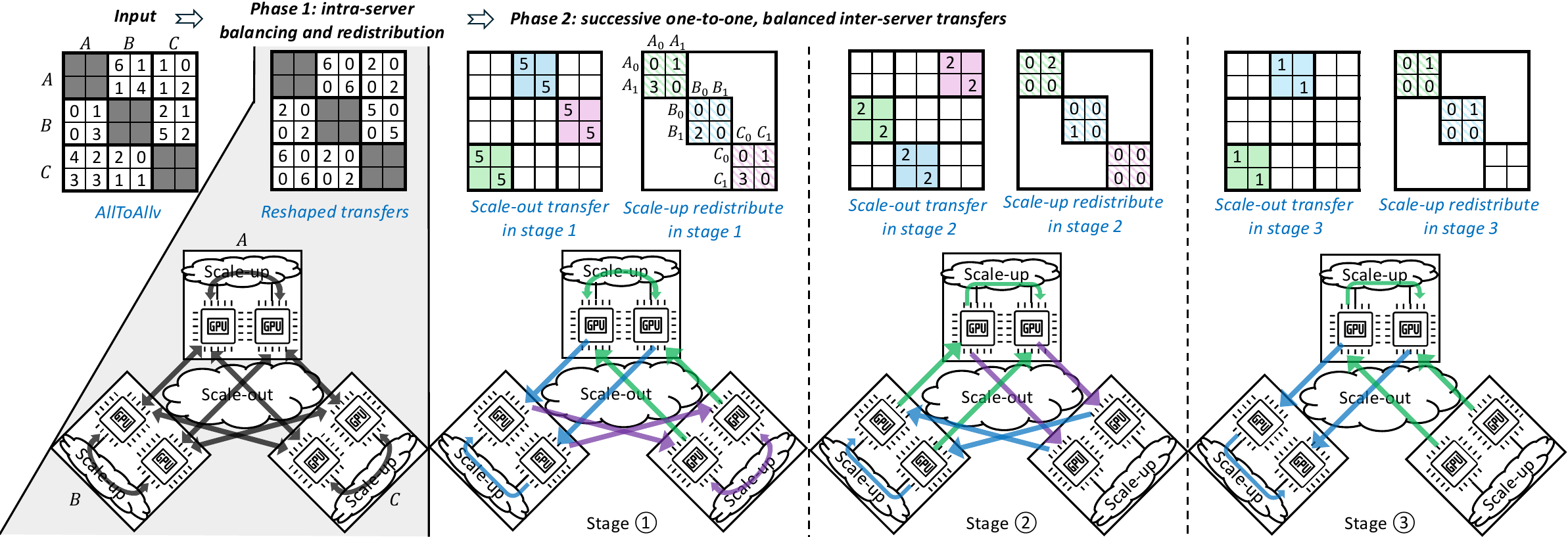}
    \caption{ 
    End-to-end scheduling example for a 3-server, 2-GPU-per-server (6×6) \atav workload. The traffic matrix (top) and transfer process (bottom) show how intra-server scheduling mitigates skew, reshaping each tile into a scalar form (left $\rightarrow$ middle).
    Inter-server scheduling then applies Birkhoff's decomposition to schedule successive 1-to-1 server transfers (middle $\rightarrow$ right). Each stage is balanced, 1-to-1, and keeps bottleneck servers active, achieving near-optimal scale-out performance.}
  \label{fig:phase2_example}
\end{figure*}

\jpara{Multi-server end-to-end scheduling}
The final piece of our scheduler---Birkhoff’s decomposition---addresses the remaining server-level skew.
We illustrate the full scheduling process with a 3-server, 2-GPU-per-server example in \Cref{fig:phase2_example}, omitting intra-server transfers (grey diagonal tiles) for clarity.
The input is a 6$\times$6 GPU-to-GPU traffic matrix (left).
Without our scheduler, the completion-time lower bound is 10 units, set by the heaviest sender GPU ($B_1$, row sum 10) and receiver GPU ($B_0$, column sum 10).

\textbf{Step~1: Balancing.}
This step reduces the severity of the bottleneck.
Within each 2$\times$2 tile (e.g., $A\rightarrow B$), sender loads are equalized across GPUs, and \emph{peer transfer} (e.g., $A_i \rightarrow B_i$) ensures receivers share the load evenly.
With intra-server skew removed, {\it the effective lower bound improves}: in the reshaped matrix (middle of \Cref{fig:phase2_example}), the maximum row/column sum drops from 10 to 8 (i.e., $A,\,C$ as sender and $A,\,B$ as receiver).
Intuitively, the pressure of a straggling NIC/GPU is averaged across all GPUs and NICs within that server, reducing its impact.
At this point, the 6$\times$6 GPU-level matrix can be cleanly collapsed into a skewed 3$\times$3 server-level \atav.

\textbf{Step 2: Balanced, one-to-one transfer stages.}
Birkhoff’s decomposition then partitions this server-level matrix into three balanced, one-to-one transfer stages (right).
Each stage delivers a portion of the workload, and together they complete all transfers.
The resulting schedule satisfies three key properties:
(i) \textit{Incast-free:} 
At the server level, Birkhoff enforces one-to-one matchings. 
Combined with Step~1’s peer-access rule, each GPU communicates {\it only with the same-index GPU in the matched server}, preventing any receiver overload. 
(ii) \textit{Balanced:} Servers send equal volumes per stage, while Step 1 guarantees balanced GPUs within each server.
(iii) \textit{Optimal:} Bottleneck servers (i.e., $A,\,C$ as senders, $A,\, B$ as receiver here) remain fully active across all stages, achieving the theoretical minimum completion time (8 units).

\textbf{Step 3: Per-stage redistribution.}
Step 2's scale-out transfers ensure data reaches the correct server, but not necessarily the correct GPU.
A lightweight redistribution step fixes placement locally, aligned with each stage.
For example, in~\cref{fig:phase2_example}, once $A\rightarrow B$ (shown as blue tile) completes in Stage \textcircled{\scriptsize 1}, the corresponding portion is immediately redistributed within $B$ (shown as the blue-striped tile).

In summary, our scheduler completes the server-level scheduling by decomposing the reshaped workload into a sequence of balanced one-to-one stages.
We now turn to how these stages are executed in practice, showing how pipelining overlaps inter- and intra-server transfers to further reduce latency and hide balancing and redistribution costs.

\subsection{End-to-End Transfer Pipeline}
\label{subsec: end_to_to_pipeline}

\begin{figure}
\centering
\includegraphics[width=\linewidth]{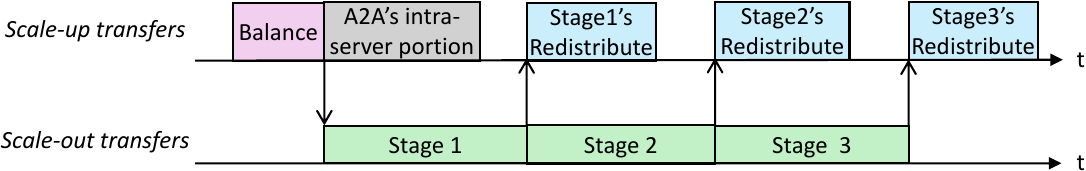}
\caption{End-to-end pipeline: scale-out transfers are kept as active as possible while scale-up operations overlap in the background. Arrows indicate triggering between transfers.}
\label{fig:end_to_end_pipeline}
\end{figure}

Our scheduler combines scale-up transfers (balancing and redistribution) with scale-out transfers (staged by Birkhoff’s decomposition).
While scale-up is much faster than scale-out, {\it serializing all the steps}, e.g.,
{\small\texttt{balance $\rightarrow$ stage~1 scale-out $\rightarrow$ stage~1 redistribute $\rightarrow$ stage~2 scale-out $\rightarrow \dots$}}, 
still results in noticeable overhead.
To minimize end-to-end completion time, we build a pipeline that keeps the scale-out tier---the true bottleneck---as busy as possible, while hiding scale-up transfers in the background. 

There are three types of scale-up transfers to consider:
(i) Balancing, which equalizes GPU loads before scale-out begins;
(ii) Redistribution, which corrects placement as scale-out completes, now aligned \emph{per stage} (e.g., the striped tiles in \Cref{fig:phase2_example}); and
(iii) Intra-server portion of \atav (grey tiles in \Cref{fig:phase2_example}).
The scale-out portion is straightforward: successive one-to-one stages from Birkhoff’s decomposition.

\jpara{Pipeline structure}
\Cref{fig:end_to_end_pipeline} illustrates how the pipeline operates. 
Balancing (purple) runs first, as all subsequent scale-out stages depend on the reshaped workload. 
Once balancing finishes, scale-out begins immediately. 
During scale-out, Stage~$i$’s redistribution (blue) overlaps with Stage~$(i+1)$’s scale-out (green), hiding redistribution cost---except in the final stage, which has no successor. 
The intra-server portion of \atav (grey) is executed alongside the first scale-out stage, making use of idle scale-up bandwidth before redistributions are triggered. 
While the pipeline could be made even tighter by subdividing balancing and scale-out into smaller chunks, the gain is small, so we adopt this simple design.

Overall, the pipeline keeps the scale-out network---the true bottleneck---as busy as possible, while scale-up operations are largely hidden in the background.
This completes the design: intra-server scheduling removes GPU-level skew, inter-server scheduling produces balanced server-level transfers, and pipelining integrates them into a seamless end-to-end execution.
Next, we turn to the key properties of our scheduler, including optimality and computational complexity.

\subsection{Scheduler Properties}
\label{subsec: properties}

Beyond the algorithm itself, the scheduler’s properties reveal why it is a practical, efficient, and high-performance solution for \atav in modern GPU clusters. Assume $N$ servers and each server has $M$ GPUs.

\jpara{Optimality}
The majority of end-to-end transfer---staged scale-out transfers generated by Birkhoff’s decomposition---operates at full efficiency on the scale-out fabric.
Suboptimality arises from additional intra-server operations such as balancing and redistribution. 
Empirically, these costs are small: under typical workloads, they add less than 5\% to the total scale-out cost (as shown in~\cref{subsubsec: perf_skew}), and even under highly skewed traffic (Zipfian with skewness 0.9), the overhead remains below 8\%.

We also prove in Appendix~\ref{subsec: worst-case-proof} that even under adversarial workloads---designed to maximize balancing (e.g., all GPUs within a server join \atav but only one holds data) and redistribution (e.g., only one GPU in each destination server receives data)---the performance gap from the theoretical optimum remains bounded.
For instance, in a 4-node cluster with a 9:1 scale-up to scale-out bandwidth ratio (450 GBps 4th-gen NVLink~\cite{h100} vs.\ 400 Gbps Ethernet), \sys finishes within 2.12$\times$ of the optimum, even in this worst-case setting.

\jpara{Number of stages}
The number of transfer stages produced by Birkhoff’s decomposition depends on the {\it server-level} matrix.
In the best case, a perfectly balanced $N \times N$ server-level matrix requires exactly $N$ stages---the same as SpreadOut.
With skew, more stages may be introduced, but the total is always bounded by $(N^2 - 2N + 2)$~\cite{birkhoff_stages}.

Fewer stages are preferable because each adds synchronization overhead.
Ideally, one would minimize the stage count, but finding such a decomposition is NP-hard~\cite{birkhoff_nphard}.
Rather than attempting this costly search, our scheduler efficiently produces a valid decomposition.
Since the stage count is bounded, the synchronization cost is also bounded---and in practice, our evaluation shows it to be negligible.

\jpara{Computational complexity \& runtime}
Intra-server scheduling is lightweight, involving only simple load balancing and bookkeeping. The primary cost lies in inter-server scheduling---Birkhoff’s decomposition---which runs in polynomial time with $O(N^5)$ complexity and is fast in practice. 

As shown in \Cref{fig:scheduling_overhead_cmp}, our scheduler completes in 25 µs for 4 servers (32 GPUs), 221 µs for 8 servers (64 GPUs), and 805 µs for 12 servers (96 GPUs), assuming $M=8$ GPUs per server---a common configuration~\cite{hgx}. 
With one expert per GPU, this spans the typical expert-parallelism (EP) range of today’s MoE workloads, from EP8---EP128 (e.g., Perplexity AI~\cite{perplexityai}) to large-scale deployments such as EP320 (DeepSeek~\cite{deepseek}). Even at EP320 (40 servers), scheduling overhead remains modest at 77 ms, well within practical bounds.

To put this in context, consider a median-scale case with EP64: each GPU transmits about 1 GB of data to others---the traffic scale reported by prior work~\cite{mixnet, deepep}.
Over a 400 Gbps network, such an \ata takes at least 20 ms, while scheduling adds 221 $\upmu$s ($\approx$1.1\% of total time).
Our scheduling step is a small upfront `tax' that yields a fully optimized plan, shortening end-to-end completion compared to no schedules.

By contrast, prior approaches~\cite{taccl, teccl, syccl} cast scheduling as NP-hard problems (e.g., MILP or multi-commodity flow) and rely on solvers~\cite{gurobi} that run orders of magnitude slower: for a 16-GPU \ata, the fastest solver-based scheduler, SyCCL~\cite{syccl}, takes 3.6 s---while our scheduler runs in  3.1 $\upmu$s.

\jpara{Exclusion of \ata scheduling over scale-up}
Both balancing and redistribution are themselves skewed \atav operations.
Because they run entirely on the fast scale-up fabric, sophisticated scheduling is unnecessary.
For these steps we use MPI’s SpreadOut algorithm~\cite{mpi_all2all}, which provides simple one-to-one sender–receiver mappings at low cost.
While Birkhoff's decomposition could deliver an optimal schedule here as well, its added computation is unnecessary---the scale-up tier is not the bottleneck.

Another caveat is that SpreadOut may not be well suited for older GPUs with non-symmetric scale-up topologies, e.g., the ring in AMD MI250~\cite{mi250x} and the hybrid cube mesh in NVIDIA V100~\cite{v100}.
However, recent GPUs adopt symmetric scale-up fabric, e.g., switch-based topology~\cite{h100} and fully connected mesh~\cite{mi300x}, which are our target platforms.

\jpara{Adapting an arbitrary matrix to a valid form}
Birkhoff’s theorem applies to \emph{scaled doubly stochastic matrices}, where all row and column sums are equal. Since real $N \times N$ server-level traffic matrices are arbitrary, we first embed them into this form by adding an auxiliary matrix, which can be constructed in $O(N^2)$ time. This procedure increases only the lighter rows or columns until all sums match the heaviest one, leaving the true bottleneck row or column unchanged.

The auxiliary entries represent \emph{virtual} transfers that are never executed and are ignored once all real traffic completes. As a result, some permutation matrices produced by the decomposition may appear partial \emph{with respect to real traffic} (\cref{fig:birkhoff_vs_spreadout}), with zero rows or columns originating from the auxiliary matrix. Importantly, this transformation preserves {\it both correctness and optimality, since the maximum row or column sum—the true bottleneck---remains unchanged}.

\jpara{Birkhoff’s decomposition: overview}
The algorithm takes as input an $N \times N$ scaled doubly stochastic matrix and outputs a sequence of permutation matrices whose weighted sum reconstructs the input. The matrix can be viewed as a bipartite graph with $N$ senders (rows) and $N$ receivers (columns), where nonzero entries correspond to edges.

At a high level, the algorithm repeatedly finds \emph{a perfect matching} in this graph; each such matching selects exactly one outgoing edge per sender and one incoming edge per receiver, yielding a permutation matrix.
Each matching can be computed, for example, using the Hungarian algorithm~\cite{Kuhn1955Hungarian} with $O(N^3)$ complexity. 
After subtracting the permutation matrix derived from each matching, the residual matrix remains scaled doubly stochastic, allowing the process to repeat. 
In the worst case, the algorithm requires $O(N^2 - 2N + 2)$ iterations~\cite{birkhoff_stages}, resulting in $O(N^5)$ total complexity.

A key advantage of the algorithm is that it advances \emph{all} bottleneck rows and columns---potentially multiple with equal maximum load---at the same rate. 
In contract, a greedy algorithm may fail to account for all bottlenecks simultaneously, often prioritizing individual large entries and suboptimal.
\section{Evaluation}
\label{sec: evaluation}

We evaluate \sys to answer four key questions:

\begin{myitemize}
\item How does \sys compare to state-of-the-art schedulers on workloads, transfer sizes, and degrees of skew (\S\ref{eval: benchmark})?
\item What end-to-end throughput gains does \sys deliver to MoE training (\S\ref{eval:endtoend})?
\item How does \sys's scheduling runtime compare with solver-based solutions (\S\ref{eval:scheduling_overhead})?
\item How does \sys scale to larger clusters and varying network bandwidths (\S\ref{eva: simulation})?
\end{myitemize}

\jpara{Testbed}
(i) {\it NVIDIA cluster:} 
4 servers with NVIDIA H200 GPUs~\cite{h200}, each with 8 GPUs, interconnected by 400 Gbps InfiniBand with credit-based flow control~\cite{infinibandflowcontrol} and 4 KB MTU. 
Intra-server scale-up uses NVLink, with a 9:1 scale-up to scale-out bandwidth ratio (450 GBps vs.\ 50 GBps). The scheduler runs on Intel Xeon Platinum 8468 CPUs.
(ii) {\it AMD cluster:} 4 servers with AMD MI300X GPUs~\cite{mi300x}, each with 8 GPUs, connected via 100 Gbps RoCEv2 Ethernet with out-of-the-box DCQCN~\cite{dcqcn} as congestion control and 1 KB MTU. 
Intra-server scale-up is a fully connected Infinity Fabric mesh, with a 35:1 bandwidth ratio (448 GBps vs.\ 12.5 GBps). The scheduler runs on AMD EPYC 9534 CPUs.
Each GPU has its dedicated NIC with GPU Direct RDMA~\cite{gpudirect} in both clusters.

\jpara{Libraries \& dependencies}
We provide a Python API, {\small \texttt{all\_to\_all\_FAST}}, mirroring PyTorch’s {\small \texttt{all\_to\_all\_single}} for integration into existing models.
Implementations of data transfers differ by hardware:
(i) On H200, scale-up/scale-out uses CUDA IPC~\cite{cuda}/NVSHMEM~\cite{nvshmem} respectively.
(ii) On MI300X, both use RCCL~\cite{rccl}.
The transfer pipeline is implemented using multiple CUDA streams with synchronizations.

\jpara{Integration into MoE systems}
\sys operates in a \emph{distributed} fashion: given the same traffic matrix, each GPU independently computes the identical global schedule, eliminating the need for a central coordinator. Only the traffic matrix---a compact integer array---must be synchronized; the schedule itself does not need to be exchanged.

This integration is natural in MoE frameworks such as Megatron-LM~\cite{megatron-lm}, which already materialize the per-\atav traffic matrix before each dispatch. Specifically, Megatron-LM performs an All-Gather of per-expert token counts (e.g., {\small \texttt{num\_global\_tokens\_per\_expert}~\cite{megatron-lm-allgather}}), from which the full traffic matrix can be constructed. 
Importantly, this All-Gather is \emph{not} introduced by \sys.
It is already required by the baseline NCCL \atav implementation to compute receive counts and buffer offsets, since each GPU knows how many tokens it sends but not how many it will receive from others.
\sys simply consumes this existing traffic matrix, synthesizes a schedule, and executes the transfers.

\jpara{Workloads}
We evaluate \sys on both synthetic and real MoE workloads.
For synthetic workloads, we model skewness by varying GPU-pair transfer sizes using two distributions: (i) {\it random} \atav with uniformly-distributed sizes, and (ii) {\it skewed} \atav with Zipfian-distributed sizes.
For real workloads, we focus on MoE training, where prior work identifies \atav as the primary bottleneck~\cite{tutel, mixnet, gshard}, and report the end-to-end throughput improvements.

We also evaluate {\it repetitive, balanced} \ata workloads, where existing schedulers can amortize their cost, enabling fair comparison in settings favoring prior approaches.

\jpara{Metrics}
Our primary metric is {\it algorithmic bandwidth}, widely used in prior work~\cite{nccl, taccl}. It captures how fast a transfer completes, defined as  $\frac{\text{Total transfer size}}{\text{\# of GPUs} \times \text{Completion Time}}$. Because skewed \atav may have variable per-GPU volumes, this average normalizes across all GPUs.
Algorithmic bandwidth can exceed the raw scale-out link bandwidth, since part of the transfer completes locally over the faster scale-up fabric. For example, in a 4-node cluster with 50 GBps scale-out links, if 25\% of the traffic is intra-server, the optimal algorithmic bandwidth is $50 / 0.75 = 66.6$ GBps. Higher is better.

Our second metric is {\it scheduling runtime}, which captures the time to synthesize a schedule. Lower is better.

\jpara{Baselines}
We compare \sys against two classes of baselines:
(i) \textit{Solver-based schedulers}: TACCL~\cite{taccl} and TE-CCL~\cite{teccl}, which used constraint-based solvers for scheduling. We evaluate them on both NVIDIA and AMD clusters.
(ii) \textit{Industry libraries}: On NVIDIA: NCCL~\cite{nccl} (version 2.27.3), DeepEP~\cite{deepep} (from DeepSeek~\cite{deepseek}), and MSCCL~\cite{msccl}. On AMD: RCCL~\cite{rccl}, SpreadOut~\cite{mpi} (abbreviated as `SPO'), and MSCCL.
Since DeepEP is NVIDIA-only, we do not evaluate it on AMD. NCCL outperforms SpreadOut on NVIDIA, so we omit SpreadOut there; on AMD, RCCL’s \atav lacks effective scheduling, making SpreadOut a stronger baseline that we include.

\subsection{\atav Performance}
\label{eval: benchmark}

\subsubsection{Performance Under Different Transfer Sizes}

\sys consistently achieves the best \atav performance on both NVIDIA and AMD testbeds.
For baselines, NCCL, DeepEP, SpreadOut, and RCCL support \atav natively, while TACCL, TE-CCL, SyCCL, and MSCCL only support balanced \ata.
Adapting the solvers to skewed \atav can be done in two ways: (i) explicitly encoding variable flow sizes into the problem formulation, or (ii) padding all flows to a uniform size so the solver sees a balanced workload (padding data is used only for scheduling, not for actual transfers).
We attempted the first approach, but it made the solvers---already slow on balanced workloads (e.g., TACCL needs over 30 minutes for 32 GPUs)---incapable of finishing within a reasonable time.
Thus, we adopt padding to simplify workloads for solver-based schedulers.

We evaluate with per-GPU message sizes from 100 MB to 1 GB, representative of typical workloads reported by prior work~\cite{mixnet, deepep}.
The improvement factors vary across testbeds, reflecting differences in both network hardware and software implementation, as discussed below.

\begin{figure}
\begin{subfigure}[b]{0.235\textwidth}
\centering
\includegraphics[width=0.95\linewidth, trim=7 1 4 4, clip]{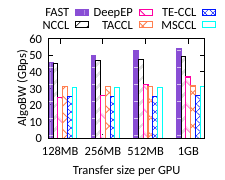}
\caption{Random}
\label{fig:nv_bw_buffer_sz_random}
\end{subfigure}
\begin{subfigure}[b]{0.235\textwidth}
\centering
\includegraphics[width=0.95\linewidth, trim=7 1 4 4, clip]{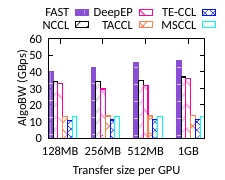}
\caption{Skewed w/ factor 0.8}
\label{fig:nv_bw_buffer_sz_skewed}
\end{subfigure}
\caption{\atav performance on NVIDIA testbed (with 450 GBps scale-up and 50 GBps scale-out).}
\label{fig:nv_bw_buffer_sz}
\end{figure}

\jpara{Result on NVIDIA testbed}
As shown in~\cref{fig:nv_bw_buffer_sz_random}, \sys achieves the best algorithmic bandwidth under random workloads. It slightly outperforms NCCL by 1.01–1.1$\times$, and exceeds DeepEP (1.5–1.9$\times$) and TACCL (1.5–1.7$\times$). Performance improves with larger transfers, as scale-out links saturate more easily and staging overheads in \sys are amortized.

NCCL with PXN~\cite{ncclpxn} employs {\it sender-side aggregation}, consolidating outgoing flows at proxy GPUs before traversing scale-out links. By aggregating traffic across flows, PXN reduces per-GPU variance and mitigates mild skew. As a result, under mildly skewed workloads, NCCL can approach \sys's performance even without explicit traffic rebalancing.

DeepEP~\cite{deepep} places aggregation and fan-out {\it on the receiver side}. 
Data are first delivered to ingress GPUs on the destination server and then forwarded via NVLink to their target GPUs. Under skew, multiple ingress GPUs may concurrently forward large volumes to the same targets, causing NVLink receive contention and local hotspots that limit throughput---as observed by DeepEP’s own NVLink runtime profiler.

Solver-based schedulers (e.g., TACCL, TE-CCL) convert skewed \atav into a fictitiously balanced \ata via padding, which significantly reduces synthesis time and enables practical schedule generation. However, the padded transfers do not correspond to real data movement and still occupy communication slots, delaying actual transfers. So their performances are much lower than \sys in practice.

Straggler effects intensify under skew. Under Zipfian workloads (\cref{fig:nv_bw_buffer_sz_skewed}), \sys outperforms NCCL by 1.2–1.3$\times$, DeepEP by 1.2–1.5$\times$, and TACCL by over 3$\times$. The performance gap with NCCL widens as the workload shifts from uniform to Zipfian: even with PXN aggregation, residual imbalance introduces stragglers that limit NCCL’s efficiency.

\begin{figure}
\begin{subfigure}[b]{0.235\textwidth}
\centering
\includegraphics[width=0.95\linewidth, trim=7 1 4 4, clip]{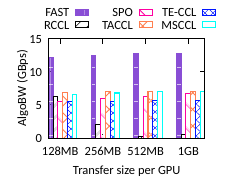}
\caption{Random}
\label{fig:amd_bw_buffer_sz_random}
\end{subfigure}
 \begin{subfigure}[b]{0.235\textwidth}
\centering
\includegraphics[width=0.95\linewidth, trim=7 1 4 4, clip]{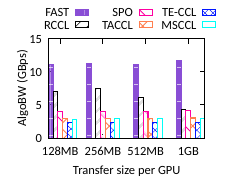}
\caption{Skewed w/ factor 0.8}
\label{fig:amd_bw_buffer_sz_zipf}
\end{subfigure}
\caption{\atav performance on AMD testbed (with 448 GBps scale-up and 12.5 GBps scale-out).}
\label{fig:amd_bw_buffer_sz}
\end{figure}

\jpara{Results on AMD testbed}
\sys again achieves the best performance under random workloads (\cref{fig:amd_bw_buffer_sz_random}), surpassing TACCL by 1.3–1.8$\times$, TE-CCL by 1.6–2.3$\times$, SpreadOut by 1.9–2.1$\times$, and RCCL by 1.1–10$\times$. As with NVIDIA, most algorithms benefit from larger transfers.

RCCL, however, shows the opposite trend: throughput decreases with transfer size. This is mainly due to its \atav implementation---launching all flows concurrently with no scheduling---causing severe incast and reduced goodput.

Under skewed workloads (\cref{fig:amd_bw_buffer_sz_zipf}), \sys extends its lead, outperforming TACCL by 2.9–3.8$\times$, TE-CCL by 3.6–4.7$\times$, SpreadOut by 2.5–2.8$\times$, and RCCL by 1.3–2.6$\times$. Interestingly, RCCL performs relatively better here than under random workloads: skew concentrates traffic into a few elephant flows while leaving most as short mice transfers, reducing widespread collisions and easing incast pressure.

\subsubsection{Performance under Balanced \ata}

On the simple, repetitive balanced workload, DeepEP (60 GBps), TACCL (59 GBps), and NCCL (58 GBps) all achieve good performance.
In this setting, \sys achieves 58 GBps---slightly below the best---since its balancing and redistribution add minor overhead unnecessary when the workload is already balanced.
While prior work efficiently handles balanced \ata, they lack mechanisms to address skew-induced stragglers, where \sys provides clear advantages.

\subsubsection{Performance under Different Skewness}
\label{subsubsec: perf_skew}

We generate skewed workloads using a Zipfian distribution with varying skewness factors. A larger factor produces more mice flows and amplifies elephant flows, creating a stronger imbalance. The \atav traces we profile during MoE pretraining show skewness factors between 0.4 and 0.8.

On the AMD testbed, we compare \sys with TACCL, SpreadOut, and RCCL across different skewness levels (\cref{fig:skewness}). 
\sys consistently delivers the best performance, outperforming RCCL by 1.6–10$\times$, SpreadOut by 2.1–3.1$\times$, and TACCL by 2.1–4.5$\times$ (TE-CCL omitted here as it performs slightly worse than TACCL).

\begin{figure}
\centering
\begin{subfigure}[b]{0.52\linewidth}
\includegraphics[width=0.95\linewidth, trim=7 1 2 4, clip]{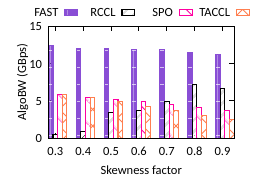}
\caption{Performance}
\label{fig:skewness}
\end{subfigure}
\begin{subfigure}[b]{0.45\linewidth}
\includegraphics[width=0.95\linewidth, trim=7 1 4 4, clip]{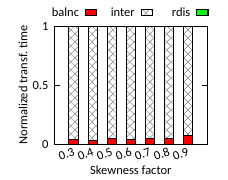}
\caption{Breakdown}
\label{fig:breakdown}
\end{subfigure}
\caption{\atav performance and transfer time breakdown under different skewness on AMD testbed.}
\label{fig:skewness_breakdown}
\end{figure}

The performance gap reflects how each system handles stragglers.
(i) For \sys, increased skew lengthens the balancing phase, but the overhead remains modest as shown in~\cref{fig:breakdown}: even at skew factor 0.9, balancing and redistribution account for under 8\% of scale-out time (under 5\% in most cases). Since scale-out dominates and runs at full efficiency, \sys remains within 1.08$\times$ of optimal.
(ii) TACCL degrades under skew because heavier skew requires more padding, reducing effective efficiency.
(iii) SpreadOut suffers as skew amplifies per-stage imbalance, causing overall transfer time to be dominated by stragglers.
(iv) RCCL shows the opposite trend: higher skew produces more mice flows, whose contention is absorbed by switch buffers, reducing incast severity and improving performance on the AMD testbed.

\subsection{End-to-End Performance}
\label{eval:endtoend}

To evaluate \sys in an end-to-end setting, we integrate it into Megatron-LM~\cite{megatron-lm} on the AMD testbed to perform on-the-fly scheduling for every \atav communication during MoE training.  
We compare against PyTorch’s~\cite{pytorch} default \texttt{\small all\_to\_all\_single} operator, which uses RCCL as the backend.  
Solver-based approaches cannot be integrated due to their prohibitive scheduling overhead.  

We vary two key MoE configurations to study how \sys behaves under different training scenarios:  
(i) \textit{Expert parallelism (EP)}: We sweep EP from 16 to 24 to 32, which directly determines the scale of \atav. Under the configuration where each GPU hosts one expert (like DeepSeek~\cite{deepseek}), this corresponds to scaling the transfer from 16 GPUs (2 servers) to 32 GPUs (4 servers).  
(ii) \textit{Top-K routing}: In MoE, each input token is routed to the Top-K most relevant experts; larger $K$ increases token replication and thus flow size in the \atav workload. For this experiment, we fix EP at 32 and vary $K$.  
Other workload-related parameters, such as batch size and sequence length, similarly affect communication size and training performance.  

\begin{figure}
\begin{subfigure}[b]{0.54\linewidth}
\centering
\includegraphics[width=0.95\linewidth, trim=7 1 4 4, clip]{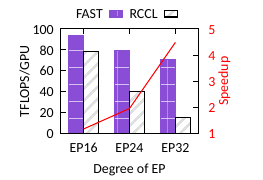}
\caption{Vary EP}
\label{fig:moe_perf_node}
\end{subfigure}
\begin{subfigure}[b]{0.43\linewidth}
\centering
\includegraphics[width=0.95\linewidth, trim=7 1 4 3, clip]{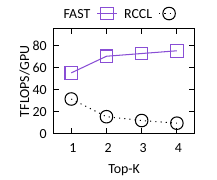}
\caption{Vary Top-K routing}
\label{fig:moe_perf_topk}
\end{subfigure}
\caption{Megatron-LM MoE training performance improvement on AMD testbed.}
\label{fig:moe_training_perf}
\end{figure}

As shown in~\cref{fig:moe_perf_node}, \sys delivers a 1.18–4.48$\times$ speedup in end-to-end training throughput (under Top-2 routing) across different EP levels. Two main trends emerge:
(i) Training throughput (left y-axis) decreases as EP increases. This is expected since higher EP involves more GPUs, more servers, and thus more scale-out traffic, which lowers communication efficiency and increases GPU idle time.
(ii) The baseline degrades sharply as EP grows due to escalating incast. 
RCCL performs no scheduling across transfers, leaving congestion entirely to the transport layer. 
For example, with EP16, a receiver GPU/NIC handles up to 8 concurrent flows, while with EP32 this rises to 24. With out-of-the-box DCQCN as congestion control, this causes severe throughput collapse.

As shown in~\cref{fig:moe_perf_topk}, \sys outperforms the baseline by 1.75–7.88$\times$. Notably, \sys and RCCL exhibit opposite scaling with $K$: increasing $K$ improves \sys by enlarging flows and amortizing staging overhead, but degrades RCCL due to increased flow collisions and congestion.

\subsection{Scheduling Overhead}
\label{eval:scheduling_overhead}

\sys introduces two types of overhead relative to non-scheduling algorithms such as NCCL:  
(i) additional scheduling runtime, and  
(ii) extra memory for intermediate buffers.  

\jpara{Scheduling runtime}  
As shown in~\cref{fig:scheduling_overhead_cmp}, \sys scales to 320 GPUs with only 77 ms of overhead---faster than SyCCL, the fastest solver-based scheduler, which already takes 3.6s at just 16 GPUs.  
Earlier solver-based methods generally fail to scale beyond 64 GPUs, rendering them unusable for moderate expert-parallelism levels such as EP96 and EP128.  

\begin{figure}
\centering
\includegraphics[width=0.75\linewidth, trim=4 1 4 4, clip]{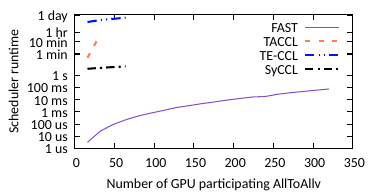}
\caption{Comparison of \sys’s scheduling runtime against state-of-the-art solver-based schedulers (log-scale).}
\label{fig:scheduling_overhead_cmp}
\end{figure}

Even at smaller scales, their scheduling time remains prohibitively long---ranging from seconds to hours---far exceeding the transfer time itself and much longer than the interval before the workload itself changes.  
In comparison, \sys’s lightweight scheduling enables {\it on-the-fly} planning.

\jpara{Memory overhead}  
\sys also requires additional memory for temporary buffers that hold rebalanced or redistributed data.  
Under random workloads, this overhead is about $\approx\mspace{-5mu}30\%$ of the original \atav buffer size.  
In practice, the impact is minimal: typical \atav buffers are under 1 GB (e.g., in DeepEP~\cite{deepep}), so the extra cost is less than 300 MB.  
On modern GPUs such as the NVIDIA H200~\cite{h200} with 141 GB of memory, this represents under 0.22\% of total capacity—an acceptable tradeoff for the performance gains.

\subsection{Scaling and Bandwidth Sensitivity}
\label{eva: simulation}

We use simulation to evaluate \sys beyond the limits of our physical testbeds, exploring both larger scales and different scale-up/scale-out bandwidth configurations.  
The simulator follows the analytical framework widely used in prior work such as TE-CCL and TACCL~\cite{taccl, teccl}: given a schedule with a sequence of transfer steps (each with a defined size), the completion time is computed by summing per-step costs.  
Each cost consists of a fixed link wake-up delay plus the transmission time ($\tfrac{\text{data size}}{\text{link bandwidth}}$).

We focus on scenarios that solver-based schedulers cannot scale to and therefore exclude them from comparison. Accordingly, we compare \sys against SpreadOut and an optimal bandwidth bound, which assumes infinitely fast scale-up links so that intra-server transfers are instantaneous. Under this bound, scale-out is the only bottleneck, and the optimal time is defined by the maximum balanced sender or receiver load divided by the scale-out bandwidth.

\jpara{Performance at larger scale}  
We first scale the number of GPUs in \atav, with each GPU pair transmitting 50 MB on average in a random workload, simulated in a 400 Gbps scale-out network and a 450 GBps scale-up network (H200).
As shown in~\cref{fig:scale_server_n}, \sys stays within 5\% of optimal when scheduling time is excluded (“FAST raw”).
Including scheduling time, the gap widens to 10\% at larger scales, since scheduling cost grows faster than workload completion time (which increases only linearly with GPU count).
We leave on-the-fly scheduling at extreme scale as future work.
By contrast, SpreadOut achieves only about half of \sys’s throughput.

\begin{figure}[t]
\centering
\begin{subfigure}[b]{0.45\linewidth}
\includegraphics[width=0.95\linewidth, trim=4 1 4 4, clip]{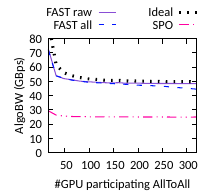}
\caption{Perf. at scale}
\label{fig:scale_server_n}
\end{subfigure}
\begin{subfigure}[b]{0.54\linewidth}
\includegraphics[width=0.95\linewidth, trim=4 1 4 4, clip]{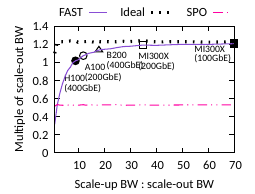}
\caption{Varying bandwidth ratios}
\label{fig:sim_speeds}
\end{subfigure}
\caption{Simulation of \sys at larger scales and varying scale-up/scale-out bandwidth ratios under random workloads.}
\label{fig:sim_scale}
\end{figure}

\jpara{Performance under varying scale-up/scale-out ratios}
We also evaluate \sys under different scale-up/scale-out bandwidth ratios on a 32-GPU setup.
As shown in~\cref{fig:sim_speeds}, normalized bandwidth is reported relative to scale-out capacity, which can exceed 1 since about 25\% of traffic is intra-server, giving an upper bound near 1.25.
Performance improves as the scale-up/scale-out ratio increases, as faster scale-up links further reduce balancing and redistribution overhead.
These results indicate that \sys sustains high efficiency on modern GPUs with advanced scale-up interconnects.
\section{Discussion}
\label{sec: discussion}
\jpara{Other collectives}
\sys is intentionally specialized for \atav, where skewness and dynamism fundamentally limit existing approaches, and is not designed as a general-purpose collective scheduler. For balanced collectives such as All-Reduce and All-Gather, communication patterns are static with uniform flow sizes, and near-optimal schedules can typically be achieved using existing library or solver-based implementations; a dynamic, traffic-aware scheduler provides little additional benefit. In practice, modern libraries (e.g., NCCL) already select among multiple collective algorithms at runtime, and \sys naturally fits this model as a specialized implementation for skewed \atav workloads.

\jpara{Other topologies}
\sys targets two-tier GPU cluster topologies that combine scale-up and scale-out fabrics, which are increasingly common in modern deployments~\cite{meta, junipter_network} with HGX-style platforms~\cite{hgx}. The core insight of \sys is to reshape traffic to reduce skew before it reaches slower network tiers, thereby simplifying the workload imposed on bottleneck fabric. While our current implementation focuses on two tiers, this principle naturally extends to multi-tier networks.

\jpara{Hybrid parallelism}
Our evaluation focuses on pure expert parallelism (EP). In hybrid deployments combining EP with tensor or pipeline parallelism, \sys can be applied directly when collectives execute in disjoint time windows. Coordinating \sys with concurrent collectives that share networks introduces additional complexity and is left for future work.

\jpara{Persistence of scale-up vs.\ scale-out bandwidth gaps}
We expect the bandwidth gap between scale-up and scale-out fabrics to persist, as scale-up interconnects benefit from short distances and tight hardware–software co-design, while scale-out networks must span longer distances and support backward compatibility and multi-tenancy. Exploiting scale-up bandwidth will therefore remain an important opportunity.

\jpara{Interaction with congestion control}
\sys operates at the collective layer and is orthogonal to transport-layer congestion control. Even under ideal congestion control, \sys can provide additional benefits by reducing traffic skew and improving NIC utilization, preventing communication from being bottlenecked by a small number of heavily loaded NICs.

\section{Related Work}
\label{sec: related}

\jpara{Collective communication schedulers}
Classic \ata algorithms such as SpreadOut~\cite{mpi_all2all} assume balanced workloads and single-tier networks. Modern GPU libraries (e.g., NCCL~\cite{nccl}, RCCL~\cite{rccl}) exploit two-tier fabrics but do not handle skew, stragglers, or incast. Solver-based schedulers such as SCCL, TACCL, TE-CCL, and SyCCL~\cite{sccl, taccl, teccl, syccl} can produce near-optimal schedules but incur prohibitive overheads, making them unsuitable for dynamic workloads. In contrast, \sys provides a scalable, on-the-fly scheduler for skewed and dynamic \atav in modern GPU clusters.

\jpara{MoE optimizations}
Prior work on mixture-of-experts training~\cite{lina, tutel, gshard, switchtransformer, deepspeed_moe, mixtral, deepseek} improves efficiency through  model design or communication–computation overlap, but typically treats \atav as a black box. \sys complements these efforts by directly optimizing \atav.

\jpara{Leveraging NIC idleness}
Recent work, FuseLink~\cite{fuselink}, exploits under-utilized NIC bandwidth but operates below the collective layer, preserving the logical communication structure. 
Although FuseLink shifts traffic onto idle NICs, it does not reassign flow pairings and thus cannot address bottlenecks from skew or incast.
\sys instead reshapes communication using overlay paths that stage transfers through intermediate GPUs, explicitly mitigating hotspots and contention.

\jpara{Switch designs}
Prior work applies Birkhoff’s decomposition to switch scheduling~\cite{birkhoff_switch1, birkhoff_switch2} and circuit-switched \ata communication~\cite{chronos}. \sys instead applies it at the GPU collective layer, running over commodity packet-switched networks without any switch changes.

\section{Conclusion}
\label{sec: conclusion}
\ata communication is critical to modern distributed systems. We present \sys, the first polynomial-time, on-the-fly scheduler for skewed and dynamic \atav. \sys absorbs skew using fast scale-up links and enforces balanced one-to-one transfers over scale-out, enabling efficient \atav communication. Across NVIDIA and AMD testbeds, \sys outperforms state-of-the-art systems while reducing synthesis time by orders of magnitude.

\section{Acknowledgments}
\label{sec: acknowlegment}

We thank our shepherd, Kai Chen, and the anonymous reviewers for their constructive feedback. We also thank Yonghao Zhuang, Zhihao Jia, Isabel Suizo, and Hugo Sadok for feedback on the drafts of this work.
This work was supported by ONR Award N000142412059 and a Sloan Research Fellowship.
Justine Sherry holds concurrent appointments as an Associate Professor at CMU and as an Amazon Scholar. This work was conducted at CMU and is not affiliated with Amazon. This work was also supported in part by ACE, one of the seven centers in JUMP 2.0, a Semiconductor Research Corporation (SRC) program sponsored by DARPA.

\bibliographystyle{acm}
\bibliography{reference}

\appendix
\section{Appendix}
\subsection{Performance Bound under Adversarial Workload}
\label{subsec: worst-case-proof}

We establish a theoretical bound on \sys's performance by analyzing its behavior under adversarial workloads that trigger its worst-case execution.

We first introduce the symbols and assumptions used for proofs.
There are $n$ servers and $m$ GPUs and $m$ NICs within each server, making a total of $m \times n$ GPUs participate in \ata.
We denote the per-GPU bandwidth of scale-up and scale-out network as $B_1$ and $B_2$.
The scale-up network topology is a switch, while the other topology's performance bound can be derived in a similar way.
The total transfer size between two different servers $i$ and $j$ is denoted as $T_{ij}$ while intra-server portion of \ata in server $i$ is denoted as $S_{i}$.
$T_{ij}$ ($i\neq j$) and $S_{i}$ constitutes the complete \ata transfer workload.
Note that $T_{ii}$ does not represent anything and is thus set to be zero for the ease of writing proofs.
We don't prove the situation where the intra-node transfer is the majority of the \ata workload because there are more scale-out pairs than scale-up pairs in multi-node \ata workload, making this scenario rare.
So, we assume each server's intra-node transfer size is no larger than the average of all the inter-node transfers, i.e.,  $S_i \leq \frac{1}{n}\sum_{j=0}^{n-1}(T_{ij})$.

\begin{theorem}
The optimal transfer completion time $t_{optimal}$ is:  
$$ \frac{1}{mB_2}\max( \max_{i=0}^{n-1}({\sum_{j=0}^{n-1} T_{ij})},\,  \max_{j=0}^{n-1}({\sum_{i=0}^{n-1} T_{ij})})$$.
\end{theorem}
\begin{proof}
Let us compute the transfer completion time in an ideal world where the bandwidth of an intra-node network is infinite.
The intra-node transfer $S_i$, load balancing, and data redistribution can therefore be instantly completed.
After load balancing, each GPU has $T_{ij}/m$ amount of data waiting to be transferred via inter-server links.
The inter-server transfers are bound by the largest senders or receivers among all the servers, which is $\max( \max_{i=0}^{n-1}({\sum_{j=0}^{n-1} \frac{T_{ij}}{m})},\,$
$\max_{j=0}^{n-1}({\sum_{i=0}^{n-1} \frac{T_{ij}}{m})})$.
So, the shortest transfer completion time is the value shown in the theorem.
Any real-world transfers must be slower or equal to this ideal transfer.
\end{proof}

\begin{theorem}
\sys's transfer worst-case completion time \( t_{FAST} \) under the adversarial workload is:
\begin{equation}
\small
\begin{aligned}
t_{FAST} &= \frac{1}{mB_2} \max \Bigg( \max_{i=0}^{n-1} \sum_{j=0}^{n-1} T_{ij}, \,
\max_{j=0}^{n-1} \sum_{i=0}^{n-1} T_{ij} \Bigg) \\
&+ \frac{m-1}{mB_1} \max_{i=0}^{n-1} \sum_{j=0}^{n-1} T_{ij} 
+ \frac{1}{nB_1} \max_{i=0}^{n-1} \sum_{j=0}^{n-1} T_{ij}  \\
&+ \frac{1}{mB_1} \max_{i,j=0}^{n-1} T_{ij}.
\end{aligned}
\end{equation}
\end{theorem}

\begin{proof}
We compute \sys's performance under adversarial workload by summing the worst-case transfer time of each transfer step:
{\small\texttt{balance $\rightarrow$ intra-server portion of \ata  $\rightarrow$ Birkhoff's stages  $\rightarrow$ Final stage's redistribution}}.

For load balancing, it would take the longest time to complete the job when $T_{ij}$ is located at a single GPU in the beginning, as it causes the largest amount of data (i.e., $\frac{m-1}{m}\cdot T_{ij}$) to be balanced.
For a specific server pair, it takes the source GPU $\frac{m-1}{m}\cdot T_{ij}\cdot \frac{1}{B_1}$ amount of time to balance data.
Among all the server pairs, this balancing step takes $t_0 = \max_{i=0}^{n-1}(\sum_{j=0}^{n-1} T_{ij}) \frac{m-1}{mB_1}$.

For the intra-server portion of \ata, the worst case is that all the $S_i$ data is moved between only two GPUs, leaving the rest of the scale-up network idle. So, the worst-case time among all the servers is $
t_1 = \max_{i=0}^{n-1}\frac{S_i}{B_1} \leq  \frac{1}{n B_1} \max_{i=0}^{n-1}(\sum_{j=0}^{n-1} T_{ij})
$ by using the assumption  $S_i \leq \frac{1}{n}\sum_{j=0}^{n-1}(T_{ij})$.

For staged inter-server transfer, Birkhoff's Theorem can generate at most $n^2 -2n + 2$ transfer steps.
We first sort them in ascending order based on each step's transfer size and get the sorted size  $l_0 \leq l_1 \leq \cdots \leq l_{n^2 -2n + 1}$.
We then execute the transfer steps in the sorted order, which successfully hides the current step's data redistribution from the next step's inter-server transfer, since
(i) the redistribution time cost of stage $i$ is $\frac{(m-1) l_i}{B_1}$ because each GPU at destination server receives $l_i$ amount of data from scale-out, all of which needs to be forwarded to a single GPU (worst case scenario);
(ii) the scale-out transfer cost of stage $i+1$ is $\frac{l_{i+1}}{B_2}$; and
(iii) $\frac{(m-1) l_i}{B_1} < \frac{l_{i+1}}{B_2}$, because $l_i \leq l_{i+1}$ and $B_1$ (e.g., 450 GBps in H100) is more than $(m-1) = 7$ times faster than $B_2$ (e.g., 50 GBps) under today's $m=8$ cluster.
This means staged scale-out transfers from Birkhoff's theorem are consecutive, making the actual scale-out transfer time $t_2$ equals $t_{optimal}$ from the ideal setting because Birkhoff let the bottleneck servers keep transmitting.

Finally, for the last stage's redistribution, since each stage does a one-to-one server matching, the worst case is when the last stage's scale-out transfer picks the largest transfer size among all the server pairs, which is $\max_{i,j=0}^{n-1} T_{ij} / m$, making the worst-case completion time as
$
t_3 = \frac{1}{B_1}\max_{i,j=0}^{n-1} \frac{T_{ij}}{m}
$.

Therefore, \sys worst-case transfer time under adversarial workload is $t_{FAST} = t_0 + t_1 + t_2 + t_3$, which is the value shown in the theorem.

\end{proof}

With optimal performance and \sys's worst-case performance, we can calculate the performance bound.

\begin{theorem}
The gap between \sys's worst-cast performance and optimal performance is bound by
$
\frac{B_2}{B_1}(m+\frac{m}{n})
$.
\end{theorem}

\begin{proof}
We divide \sys worst-case transfer time by ideal transfer time as follows:
$
\frac{t_{FAST}}{t_{optimal}} = \frac{t_0 + t_1 + t_2 +t_3}{t_{optimal}} \leq 1 + \frac{B_2}{B_1}(m+\frac{m}{n})
$
where we shrink the denominator as follows: 
$$
\max( \max_{i=0}^{n-1}({\sum_{j=0}^{n-1} T_{ij})},\,  \max_{j=0}^{n-1}({\sum_{i=0}^{n-1} T_{ij})}) \geq \max_{i=0}^{n-1}({\sum_{j=0}^{n-1} T_{ij})} \geq \max_{i,j=0}^{n-1} T_{ij}
$$
to cancel out the numerator and get the final result.
\end{proof}

In conclusion, under adversarial workloads, the worst-case performance gap of \sys relative to optimal is bounded by the scale-up to scale-out bandwidth ratio.
With today’s hardware—for example, a 4-node cluster with 450 GBps scale-up on H100~\cite{h100} and 400 Gbps scale-out—this bound implies that \sys's worst-case scenario completes within 2.12× of the theoretical optimum.
In practice, this worst-case adversarial workload rarely happens and the performance is much closer to optimal as we show in the evaluation.

\end{document}